\documentclass[a4paper,10pt]{article} 

\usepackage{preamble}

\title{On generating direct powers of dynamical Lie algebras}

\author[1]{Jonathan Allcock  \thanks{jonallcock@tencent.com}}
\author[2,3]{Miklos Santha \thanks{cqtms@nus.edu.sg}}
\author[1]{Pei Yuan  \thanks{peiyuan@tencent.com}}
\author[1]{Shengyu Zhang \thanks{shengyzhang@tencent.com}}
\affil[1]{Tencent Quantum Laboratory}
\affil[2]{CQT, National University of Singapore} 
\affil[3]{CNRS, IRIF, Université de Paris}
\date{}

\begin{document}

\maketitle

\abstract{The expressibility and trainability of parameterized quantum circuits has been shown to be intimately related to their associated dynamical Lie algebras (DLAs).  From a quantum algorithm design perspective, given a set $\+A$ of DLA generators, two natural questions arise: (i) what is the DLA $\.g_\+A$ generated by $\+A$; and (ii) how does modifying the generator set lead to changes in the resulting DLA. While the first question has been the subject of significant attention, much less has been done regarding the second. In this work we focus on the second question, and show how modifying $\+A$ can result in a generator set $\+A'$ such that $\.g_{\+A'}\cong \bigoplus_{j=1}^{K}\.g_\+A$, for some $K \ge 1$.  In other words, one generates the direct sum of $K$ copies of the original DLA. 
 In particular, we give qubit- and parameter-efficient ways of achieving this, using only $\log K$ additional qubits, and only a constant factor increase in the number of DLA generators.  For cyclic DLAs, which include Pauli DLAs and QAOA-MaxCut DLAs {as special cases}, this can be done with $\log K $ additional qubits and the same number of DLA generators as $\+A$.

}



\section{Introduction}

\subsection{Background}

Variational quantum algorithms are based on parameterized quantum circuits (PQCs) consisting of tunable gates from some finite set. By adjusting the parameters to optimize a given loss function, a wide variety of problems can, in principle, be solved via this framework.  However, in practice, the quality of solution depends on the expressibility of the PQC, i.e., the set of states that the PQC can output (for a given input state), while the efficiency with which a good solution can be found depends on the ease with which the parameters can be optimized.  In certain cases, trainability can be critically hindered by the presence of so-called barren plateaus: regions of parameter space where the loss function becomes exponentially flat as the problem size increases. 

Recently, an important connection was made between the properties of variational quantum algorithms and their associated dynamical Lie algebras (DLAs)~\cite{ragone2024lie, fontana2024characterizing}.  At a high level, DLAs with large dimensions (e.g. those growing exponentially with the number of qubits) favor expressibility, but lead to barren plateaus. Conversely, barren plateaus may be avoided if the DLA dimension grows only polynomially with qubit number, but at the cost of restricting the space of states expressible by the PQC.   The dimension of DLAs has also been shown to upper-bound the number of parameters required for overparameterization to occur~\cite{larocca2023theory}, where the PQC can explore all relevant directions in state space and which  facilitates fast convergence to global optima.

The analysis of DLAs is therefore of central importance to understanding the feasibility of variational quantum algorithms and to aiding in their design.  In this regard, given a set $\+A$ of DLA generators (corresponding to the set of tunable gates in a PQC), two natural questions arise:  
\begin{enumerate}[(i)]
    \item What is the DLA $\.g_\+A$ generated by $\+A$?
    \item How does modifying $\+A$ lead to changes in the resulting DLA?
\end{enumerate}
While answers to the first question provide useful information about the expressibility and trainability of an existing PQC, answers to the second can help in constructing new PQCs with desirable properties from old ones. 

The first question has been the subject of significant attention. Unfortunately, in general, it can be difficult to analytically determine $\.g_\+A$ from $\+A$, while numerically computing properties of $\.g_\+A$ may be inefficient as $\dim(\.g_\+A)$ 
can be exponential in the number of qubits $n$.  That said, there are cases where DLAs can be analytically deduced. For instance, it is known that if $\+A$ consists of two randomly chosen generators then, with probability one, $\.g_\+A=\.{su}(2^n)$~\cite{jurdjevic1997geometric}.  Another situation amenable to analysis is that of Pauli Lie algebras, where each generator is proportional to an $n$-qubit Pauli operator. Aguilar et al.~\cite{aguilar2024full} have shown how any Pauli DLA falls into one of four families, and gave an efficient method for determining a given DLA from the anticommutation properties of its generators.  Multi-angle QAOA-MaxCut DLAs have also been studied. A special case of Pauli DLAs, these correspond to a modified quantum approximate optimization algorithm (QAOA)~\cite{farhi2014quantum} for solving the graph MaxCut problem, and were recently fully classified by Kokcu et al.~\cite{kokcu2024classification} and Kazi et al.~\cite{kazi2024analyzing}.  Standard QAOA-MaxCut DLAs, corresponding to the original QAOA algorithm and where each generator is a sum of $n$-qubit Pauli operators, were studied in~\cite{allcock2024dynamical,d2025controllability} where, among other things, the DLA for the $n$-vertex cycle graph was identified and found to be isomorphic to $\bigoplus_{j=1}^{n-1}\.{su}(2)$. The QAOA-MaxCut DLA for the complete graph was also identified, and an explicit basis given \cite{allcock2024dynamical}.

Less work has been done regarding the second question. Of note is the work by Zimborás et al.~\cite{zimboras2015symmetry} who investigated this from the lens of simulability. That is, given sets $\+A$ and $\+B$, they gave necessary and sufficient conditions for $\.g_\+A=\.g_{\+A \cup \+B}$, and thus for the system with interactions in $\+A$ to simulate interactions from $\+B$.



The second question above is the focus of this work, where we investigate modifications to DLA generators which lead to the resulting DLA being a direct power of the original DLA. That is, if the generating set $\+A$ for DLA $\.g_\+A$ is modified to become $\+A'$, then $\.g_{\+A'}=\bigoplus_{j=1}^K \.g_\+A$, for some integer $K \ge 2$. Our motivations for doing so are threefold:
\begin{enumerate}
    \item \textbf{A tool to aid DLA analysis.} Direct power DLA structures naturally appear in various situations, e.g.\ in the DLA corresponding to cycle graphs in the quantum approximate optimization algorithm for MaxCut~\cite{allcock2024dynamical, d2025controllability}, or when the DLA generators are strings of Pauli operators~\cite{aguilar2024full}. The analysis of these DLAs can be involved, and a question is whether, in some cases, their structure can be understood as a modification of a simpler DLA. 
    \item \textbf{Constructing DLAs that avoid barren plateaus.}  From a quantum control theory perspective, one often seeks \textit{full controllability} of a system which, for $n$ qubits, corresponds to a DLA equal to $\.{su}(2^n)$. However, for variational quantum algorithms, full controllability can be problematic. Consider a variational algorithm with initial state $\rho$, measurement operator $O$, and vector $\theta$ of tunable parameters. The corresponding DLA $\.g_\+A$ decomposes into direct sum of simple subalgebras $\.g_1,...,\.g_k$ (for $k\ge 1$) and a center $\.c$, and it is known~\cite{ragone2023unified} that if $\rho\in i\.g_\+A$ or $O\in i\.g_\+A$,  then the variance of the loss function $\ell(\rho,O,\theta)$ satisfies 
    \begin{equation*}
    \var_\theta[\ell(\rho,O;\theta)] = \sum_{j=1}^k \frac{\+P_{\.g_j}(\rho) \+P_{\.g_j}(O)}{\dim(\.g_j)},
\end{equation*}
provided the circuit is deep enough to form a $2$-design. Here $\+P_{\.s}(H)\defeq \sum_{j=1}^{\dim(\.s)}|\tr(E_j^\dag H)|^2$, and $\{E_j\}_{j=1}^{\dim(\.s)}$ is an orthonormal basis for a Lie algebra $\.s$. Avoiding barren plateaus requires, by definition, that $\var_\theta[\ell(\rho,O;\theta)] =\Omega(1/\poly(n))$. However, for $\.{su}(2^n)$, which is a simple Lie algebra (and hence has no non-trivial simple subalgebras) of dimension $4^n-1$, the variance vanishes exponentially quickly in $n$.  The ability to construct DLAs which decompose into multiple copies of smaller subalgebras can therefore be used as a tool in designing algorithms which avoid barren plateaus.
    
    \item \textbf{Finding qubit- and parameter-efficient ways of constructing DLAs.} Given a set of generators $\+A$ on $n$ qubits, it is trivial to define a DLA $\+A'$ on $nK$ qubits such that $\.g_{\+A'}=\bigoplus_{j=1}^K \.g_\+A$: simply let $\+A'$ consist of $K$ copies of $\+A$ acting on disjoint sets of qubits.  However
    , in addition to the large number of qubits required, the cardinality of the modified set is $\abs{\+A'}=K\abs{\+A}$. As the number of tunable parameters in a PQC is lower-bounded by $\abs{\+A}$, the variational optimization of such a system requires tuning at least $K\abs{\+A}$ parameters. On the other hand, the number of tunable parameters is only upper-bounded by the total number of gates in the PQC, which is independent of the DLA.  Reducing the number of qubits required is clearly desirable from a quantum algorithm perspective, as is maintaining the minimum number of tunable parameters required for overparameterization.   
\end{enumerate}

\subsection{Dynamical Lie algebras}
  A Lie algebra $\.g$ is a vector space over a field $\mb{F}$ equipped with a bilinear map $[\cdot,\cdot]$ called the \textit{Lie bracket} that satisfies, for all $A,B,C\in\.g$, (i) $[A,A]=0$ and (ii) $[A,[B,C]] + [B,[C,A]]+[C,[A,B]]=0$ (Jacobi identity). We call a set of linearly independent and traceless anti-Hermitian operators on $n$ qubits a {\it dynamical generating set}. For an integer $L \geq 2$ and a dynamical generating set $\+A = \{A_1, \ldots, A_L\}$, the  \textit{dynamical Lie algebra} (DLA) generated by $\+A$, denoted
\eq{
\.g_{\+A} &= \langle \+A \rangle_{\text{Lie}, \R},
}
is the smallest Lie algebra over $\mb{R}$ containing  $\+A$ and closed by the Lie bracket, which is taken to be the commutator $[A_1,A_2] \defeq A_1A_2-A_2A_1$.
For $s \ge 2$, define right-nested commutators of elements of $\+A$ to be nested commutator of
the form $[A_{i_1} , [\ldots [A_{i_{s-1}} , A_{i_s} ]\ldots]]$.

The Jacobi identity implies the following well-known fact, the proof of which is standard, but which, for completeness, we include details of in Appendix~\ref{app:nested-comms}.

\begin{fact}\label{fact:nested} All nested commutators of elements in $\+A$ can be expressed as linear combinations of right-nested commutators.
\end{fact}



Given a DLA $\.g_\+A$, its \textit{semisimple component}, or \textit{commutator subalgebra}, is defined as $[\.g_\+A,\.g_\+A]\defeq \spn\{[A,B] : A,B\in\.g_\+A\}$, and its \textit{center}
is defined to be $Z(\.g_\+A)\defeq\{A\in\.g_\+A : [A,B]=0 \quad \forall B\in \.g_\+A \}$. 
As we will deal only with real Lie algebras, when we use the terms `linear independence' and `span' these will always mean over $\mb{R}$.  

\begin{fact}[\cite{knapp1996lie}]\label{directsum}
  As a subalgebra of $\.{su}(2^n)$, every dynamical Lie algebra admits a decomposition 
    \eq{
    \.g_\+A = [\.g_\+A,\.g_\+A] \oplus Z(\.g_\+A).
    }
A Lie algebra admitting such a decomposition is called reductive.
\end{fact}

We say that $\+A$ is a {\it Pauli dynamical generating set} if, for all $j\in[L]$, $A_j = iP_j$ where $P_j$ is an $n$-qubit Pauli operator (i.e., an $n$-fold tensor product of the single qubit Pauli operators $I,X,Y,Z$), and we call the corresponding DLA a \textit{Pauli DLA}.

\section{Our Results}
We consider two ways of modifying a DLA generating set. In the first (cardinality-increasing modifications), the number of generators in the modified set is larger than the original generating set.

\begin{restatable}{thm}{ExtendGenSetNaive}\label{thm:extendgensetnaive}
Let  $\+A = \{A_1, \ldots, A_L\}$ be a dynamical generating set, and $\chi$ a Hermitian operator with $K$ distinct  eigenvalues. Define $\+A' = \{A_i \otimes \chi^j : i\in[L], 0\le j\le K-1\}$. Then, $\.g_{\+A'}   \cong \bigoplus_{j=1}^{K}\.g_\+A$.
\end{restatable}
While the cardinality of the modified set is $K\abs{\+A}$ -- the same as the naive approach mentioned in the introduction -- only at most $\lceil \log K \rceil$ additional qubits are required to generate a DLA that is isomorphic to a direct sum of $K$ copies of $\.g_\+A$. This is because it is possible to define a Hermitian operator $\chi$ on $ \lceil \log K \rceil$ qubits that has any $K \ge 1$ distinct  eigenvalues.  The factor $K$ increase in the number of generators may seem large.  However, for any $\+A$, $\dim(Z(\.g_{\+A}))\le \abs{\+A}$~\cite{allcock2024dynamical}, and this bound is tight. Thus, in the worst case, this multiplicative factor of $K$ is necessary.  Reducing the number of generators to fewer than $K\abs{\+A}$ therefore requires sacrificing copies of the center, and our remaining results focus on generating direct powers of the commutator subalgebra $[\.g_\+A,\.g_\+A]$ instead of generating direct powers of $\.g_{\+A}$. 


Let $V$ be a finite dimensional Hilbert space, $\chi:V\ra V$ a Hermitian operator and let
$I$ denote the identity operator on $V$. For $\+A = \{A_1, \ldots, A_L\}$ and $S \subseteq [L],$ we denote by $\+A_S$ the set of operators 
$\{A_i \otimes I : i \in [L] \} \cup \{A_i \otimes \chi : i \in S\}$. 
To simplify notation, for $1 \le i \leq L$,  we write $\+A_i$ for $\+A_{\{i\}}$.

\begin{restatable}{thm}{ExtendGenSet}\label{thm:extendgenset}
Let  $\+A = \{A_1, \ldots, A_L\}$ be a dynamical generating set, and $\chi$ a Hermitian operator with $K$ distinct eigenvalues. Then,
 \eq{
  [\.g_{\+A_{[L]}}, \.g_{\+A_{[L]}}]  \cong \bigoplus_{j=1}^{K}[\.g_\+A,\.g_\+A] \text{ ~~ and ~~ } 
  \dim(Z(\.g_{\+A_{[L]}})) = 2\dim(Z(\.g_\+A)).
  }
\end{restatable}

By doubling the number of generators to $2\abs{\+A}$ and by adding only $\lceil\log K\rceil$ qubits, it is therefore possible to generate a DLA whose semisimple component is isomorphic to the direct sum of $K$ copies of  $[\.g_\+A,\.g_\+A]$.

We are able to show slightly stronger results for two special cases. In both, the cardinality of the modified generating set need 
only be one larger than the original set $\+A$. The first case is when the elements of $\+A$ are all $n$-qubit Pauli strings:

\begin{restatable}{thm}{PauliGenSet}\label{thm:pauli}
For every Pauli dynamical generating set $\+A = \{A_1, \ldots, A_L\}$ 
    whose anticommutation
   graph is connected, we have, for every $1\leq i \leq L$, 
    \eq{ 
   \.g_{\+A_{i}} \cong \bigoplus_{j=1}^{K}\.g_\+A.
    }
\end{restatable}
This theorem can be seen as generalizing part of the result of Aguilar et al.~\cite{aguilar2024full}. There, the authors showed that all Pauli dynamical generating sets $\+A$ satisfy $\.g_{\+A}\cong\.g_{\+B'}=\bigoplus_{i=1}^{2^{n_c}}\.g_{\+B}$, where $n_c\in\mb{Z}_+$ (the number of `control qubits') and $\+B'=\+B \cup \+C$, where  $\+B=\{B_1,..., B_{L-n_c}\}$ and $\+C=\{C_1,..,C_{n_c}\}$ are disjoint Pauli generating sets. The elements of $\+B'$ are Pauli strings on $n+n_c$ qubits, for some integer $n$, where  each $B\in \+B$ has support on the first $n$ qubits, and each  $C_j\in\+C$ has the form $C_j= B_2\otimes P_j$, where $P_j$ has support on qubit $n+j$. In other words, $\.g_\+B'$ can be seen as the DLA that results from repeatedly applying Theorem~\ref{thm:pauli} to $\+B$, each time adding $B_2\otimes P_j$ to the generating set (i.e., setting $\chi=P_j$) for $j=1,..., n_c$. As Pauli operators have eigenvalues $\pm 1$, this leads to a doubling of the DLA each time.

The second special case is when $\+A$ has cardinality $2$:

\begin{restatable}{thm}{CardTwoGenSet}\label{thm:twogenerators} For every dynamical generating set $\+A = \{A_1, A_2\}$ we have, for $i \in \{1,2\}$,
    \eq{
    [\.g_{\+A_{i}}, \.g_{\+A_{i}}] 
    \cong \bigoplus_{j=1}^K [\.g_\+A, \.g_\+A]
    }
and
   \eq{
\dim( Z(\.g_{\+A_{i}})) &=
\begin{cases}
\dim( Z(\.g_\+A)) &  \text{if 
$A_{i} \in [\.g_\+A, \.g_\+A]$,} \\
\dim( Z(\.g_\+A)) +1 & \text{otherwise.}
\end{cases}
}
\end{restatable}

Our second way of modifying a DLA generating set (cardinality-preserving modifications) does so without increasing the number of generators, but applies only to certain DLAs that satisfy a property we call cyclic. We defer a technical definition of this property to Section~\ref{sec:gensetcardincrease}, but note here that cyclic DLAs include all QAOA-MaxCut DLAs, all Pauli DLAs and all two-generator DLAs where one generator $A$ satisfies $A^2 = \lambda I$ for some $\lambda < 1$. While for cardinality-increasing modifications we introduce a Hermitan operator $\chi$, in this case we require a Hermitian operator $Q$ that is sign unambiguous, i.e., if $\lambda\neq 0$ is an eigenvalue of $Q$, then $-\lambda$ is not also an eigenvalue.

\begin{restatable}{thm}{CardGenSetFixed}\label{thm:comcyc}
Let $\+A = \{A_1, \ldots, A_L\}$ be a cyclic dynamical generating set,
and $Q$ a sign unambiguous
 Hermitian operator with $K$ distinct non-zero eigenvalues.
Set $\+A_{Q}=\{A_1\otimes Q, \ldots, A_L\otimes Q\}$. Then, 
\eq{
[\.g_{\+A_{ Q}},\.g_{\+A_{ Q}}]\cong \bigoplus_{j=1}^{K} [\.g_{\+A},\.g_{\+A}]
,
}
and 
\eq{
 Z(\.g_{\+A_{Q}}) =   \{C\otimes Q : C \in Z(\.g_{\+A})\}.
}   
\end{restatable}
In this case, the modified generating set $\+A_Q$ has cardinality $\abs{\+A_Q}=\abs{\+A}$, and only an additional $\lceil \log K \rceil$ qubits are required to generate the $K$-fold direct sum of $[\.g_\+A,\.g_\+A]$.








 \section{Preliminaries}
In this section we introduce  some terminology, and give a number of preliminary results needed in the sequel. All operators we consider are assumed to act on finite dimensional Hilbert spaces.

For a positive integer $n$ we denote the set $\{1, \ldots, n\}$ by $[n]$. 
For two vector spaces $V_1$ and $V_2$ with bases $B_1$ and $B_2$, respectively, their tensor product is $V_1\otimes V_2 \defeq \spn\{v_1\otimes v_2: v_1\in B_1, v_2\in B_2\}$. A linear map $\phi : \.g \rightarrow \.h$ from a Lie algebra $\.g$ to a
Lie algebra $\.h$ is a {\it Lie algebra homomorphism} if
\eq{
[ \phi(U), \phi(U')] = \phi([U,U']),
}
for all $U,U' \in \.g.$
A {\it Lie algebra isomorphism} is a Lie algebra homomorphism which is also
a vector space isomorphism.
For isomorphic Lie algebras $\.g$ and $\.h$ we will use the notation
$\.g \cong \.h.$  We say that the linear operators $A$ and $B$ are {\it proportional} if there exists a real number $c \neq 0$ such that $A = cB.$ For proportional operators we use the notation $A \propto B$.
 For two sets of linear operators $\+S$ and $\+T$, we define their tensor product as $\+S \otimes \+T = \{S \otimes T : S \in \+S, T \in \+T\}$.


\begin{lem}\label{decompositions}
    For any dynamical generating set $\+A$, every $U \in \.g_\+A$ 
    has a decomposition
    \eq{
     U = A + W,
    }
    where $A \in \spn( \+A)$ and $W \in [\.g_\+A,\.g_\+A].$
\end{lem}
\begin{proof}
Follows directly from the fact that $\.g_\+A \subseteq \spn(\+A \cup [\.g_\+A,\.g_\+A])$.
\end{proof}

Given a dynamical generating set $\+A$, we define an $\+A$-{\it sequence}
to be a sequence 
\eq{
\underline{A} =(A_{k_{1}},  A_{k_{2}}, \ldots , A_{k_{\ell-1}}, A_{k_{\ell}} )
}
of generators where $\ell \geq 0$ and $k_1, \ldots, k_{\ell} \in [L]$.
We call $\ell$ the {\it length} of ${\underline{A}},$  and we say that
${\underline{A}}$ is {\it proper} if $\ell \geq 2.$
For a proper $\+A$-{\it sequence} $\underline{A}$,
the {\it start} and the {\it end} of 
$\underline{A}$ are respectively
defined as the $\+A$-sequences $\start(\underline{A}) = (A_{k_{1}},  A_{k_{2}})$
and $\fin(\underline{A}) = (A_{k_{3}}, \ldots , A_{k_{\ell-1}}, A_{k_{\ell}}),$
and the {\it value} 
of ${\underline{A}}$ is the right-nested commutator
 \eq{ \val({\underline{A}}) =
[A_{k_{\ell}}, [ A_{k_{\ell - 1}}, \ldots , [A_{k_{2}}, A_{k_{1}}]\ldots]].
}
If $\underline{A'} = (A_{k'_{1}}, \ldots ,  A_{k'_{\ell'}})$ is also an
$\+A$-{\it sequence} 
then the {\it concatenation} of $\underline{A}$ and
$\underline{A'}$ is defined as
\eq{
\underline{A} \circ \underline{A'} = (A_{k_{1}}, \ldots ,  A_{k_{\ell}},
A_{k'_{1}}, \ldots ,  A_{k'_{\ell'}}).
}
We say that $\underline{A} \circ \underline{A'}$ is an \textit{extension} of $\underline{A}$
if $\ell' \geq 1,$ 
and it is a {\it stable} extension if $ \val({\underline{A}} \circ \underline{A'}) \propto  \val({\underline{A}}).$
 We say that a  basis $\+V = \{V_1, \ldots, V_D\}$ 
     for $ [\.g_\+A,\.g_\+A]$ is a {\it right-nested $\+A$-commutator basis} if every $V\in\+V$ is the value of a proper $\+A$-sequence. 
Observe that the concatenation of $\+A$-sequences is an associative operation.

From Fact~\ref{fact:nested}, the following is immediate:
\begin{prop}  
\label{prop:generating-basis}
    Every DLA has a right-nested $\+A$-commutator basis for $[\.g_\+A,\.g_\+A]$. 
\end{prop}

From hereon, we denote by $\chi$ a non-zero Hermitian operator, and by $K$ the number of distinct eigenvalues of $\chi$.  For each $j\in[K]$, we denote by $\Pi_j$ the orthogonal projector onto the $\lambda_j$ eigenspace of $\chi$. These satisfy $\Pi_i^2 = \Pi_i$ for all $i\in [K]$, and $\Pi_i \Pi_j = 0$ for all $i\neq j$. Finally, let $\+U_\chi = \{\chi^j : 0 \leq j \leq K-1\}$
 and $\+P_\chi = \{\Pi_j : j\in [K]\}$.

\begin{lem}\label{lem:projectors}
 \begin{enumerate}[(i)]
 \item[]
 For every $\chi=\sum_{k=1}^K \lambda_k \Pi_k$ as described above: 
 \item the set $\+U_\chi  = \{\chi^0, \chi^1, \ldots, \chi^{K-1}\}$ is linearly independent;\label{line:projectors-1}
 \item $\spn(\{\Pi_1, \ldots, \Pi_K\}) = \spn (\+U_\chi)$; \label{line:projectors-2}
 \item $\spn\{\chi^j : j \ge 0\}=\spn(\+U_\chi)$. \label{line:projectors-3}
 \end{enumerate}

\end{lem}
\begin{proof}
$\chi$ has  spectral decomposition $\chi = \sum_{i=1}^K \lambda_i \Pi_i$, and 
we can express the successive powers $\chi^0, \chi, \ldots, \chi^{K-1}$ as the matrix equation
\eq{
\bpm \chi^0 \\ \chi^1 \\ \vdots \\ \chi^{K-1}\epm &= 
\bpm 
1 & 1 & \cdots & 1 \\
\lambda_1 & \lambda_2 & \cdots & \lambda_K \\
\vdots & \vdots & \ddots & \vdots \\
\lambda_1^{K-1} & \lambda_2^{K-1} & \cdots & \lambda_K^{K-1}
\epm
\bpm
\Pi_1 \\ \Pi_2 \\ \vdots \\\Pi_K
\epm.
}
The set of projectors $\{\Pi_i\}_{i=1}^K$ is linearly independent. As the eigenvalues are distinct, the Vandermonde matrix can be inverted. This proves \eqref{line:projectors-1} and \eqref{line:projectors-2}.  To finish, \eqref{line:projectors-3} then follows immediately from the fact that $\chi^k = \sum_{i=1}^K \lambda_i^k \Pi_i $, and thus  $\chi^k\in \spn(\{\Pi_1, \ldots, \Pi_K\})$, for every $k\ge K$.
\end{proof}


 \begin{lem} \label{lem:copies_from_projectors}
Let $\{U_1, \ldots, U_{\delta} \}$ be a basis for a dynamical Lie algebra $\.g$. 
Then, $\bigcup_{j=1}^K\{U_1\otimes \Pi_j, \ldots, U_{\delta}\otimes\Pi_j \}$ is a basis for a dynamical Lie algebra 
isomorphic to $\bigoplus_{j=1}^K\.g$. A basis for the $j$-th copy of $\.g$ is $\{U_1\otimes \Pi_j, \ldots, U_{\delta}\otimes \Pi_j\}$.  
\end{lem}
\begin{proof} Follows directly from the fact that, for all $p,q\in[\delta]$, $[U_p\otimes \Pi_j, U_q\otimes \Pi_j] = [U_p,U_q]\otimes \Pi_j$, and $[U_p\otimes\Pi_i, U_q \otimes \Pi_j] = 0$ for $i\neq j$.
\end{proof}

 Let $\+A = \{A_1, \ldots, A_L\}$ be a 
dynamical generating set. Recall that for any $S \subseteq [L],$ we denote by $\+A_S$ the set of operators
$\{A_i \otimes I : i \in [L] \} \cup \{A_i \otimes \chi : i \in S\}$
and, 
we write $\+A_i$ for $\+A_{\{i\}}$. 
We define $\+A_\chi$ as $\{A_i \otimes \chi : i \in [L]\}.$
\begin{lem}\label{lem:center-from-original} 
    For all $S \subseteq [L],$  and for $\.g_{\+A_{\chi}} $ we have 
    \begin{enumerate}[(i)]
    \item 
    $\.g_{\+A_{S}}, \.g_{\+A_{\chi}} \subseteq \.g_{\+A} \otimes \spn(\+U_\chi),$ \label{line:center-from-original-1} 
    \item $[\.g_{\+A_{S}}, \.g_{\+A_{S}}], [\.g_{\+A_{\chi}}, \.g_{\+A_{\chi}}]  \subseteq [\.g_{\+A}, \.g_{\+A}] \otimes \spn(\+U_\chi),$ \label{line:center-from-original-2} 
    
    \item $Z(\.g_{\+A_{S}}), Z(\.g_{\+A_{\chi}}) \subseteq Z(\.g_\+A) \otimes \spn(\+U_\chi).$ \label{line:center-from-original-3} 
\end{enumerate}
\end{lem}

 \begin{proof}
 We prove the statements for $\.g_{\+A_{S}}$; the proofs for $\.g_{\+A_{\chi}}$ are identical.
 The proof of \eqref{line:center-from-original-1} is by definition. 
 For \eqref{line:center-from-original-2}, observe that if $U_1, U_2 \in \.g_{\+A}$ and
 $T_1, T_2 \in \spn(\+U_\chi),$ then 
 $[U_1 \otimes T_1, U_2 \otimes T_2] = [U_1, U_2] \otimes T_1T_2 \in [\.g_{\+A}, \.g_{\+A}] \otimes \spn(\+U_\chi)$ by Lemma \ref{lem:projectors} \eqref{line:projectors-3}.
 For \eqref{line:center-from-original-3}, 
let $C \in Z(\.g_{\+A_{S}}).$ By \eqref{line:center-from-original-1}
it can be decomposed as $C = \sum_{j=0}^{K-1} C_j \otimes \chi^j,$
where $C_j \in \.g_{\+A}$, for $0 \leq j \leq K-1.$ Since 
$C \in Z(\.g_{\+A_{S}}),$ for every element $U \in \.g_{\+A},$ we have
\eq{
0 = [C,U \otimes \chi^0] = \sum_{j=0}^{K-1} [C_j,U] \otimes \chi^j.
}
By Lemma \ref{lem:projectors} \eqref{line:projectors-1}, these $\chi^j$ are linearly independent. This implies that $[C_j,U]=0$ for all $U \in \.g_{\+A}$, and therefore $C_j \in Z(\.g_{\+A})$,
for $0 \leq j \leq K-1.$
\end{proof}

\suppress{ 
\begin{proof} Let $\+V=\{V_1,\ldots, V_D\}$ be a largest cardinality set of linearly independent nested commutators of $\+A$, and let $\tilde{\+A}\subseteq \+A$ be a largest cardinality subset of $\+A$ such that $\tilde{\+A}\cup\+V$ is linearly independent. Then $\tilde{\+A}\cup\+V$ is a basis for $\.g_\+A$, and $[A, V]\in \spn(\+V)$ for all $A\in\+A$, $V\in \+V$.  Furthermore, by the maximality of $\+V$, $[A,A']\in\spn(\+V)$ for all $A,A'\in\tilde{\+A}$. Clearly $\spn(\+V)\subseteq [\.g_\+A,\.g_\+A]$. We will show that $[\.g_\+A,\.g_\+A]\subseteq \spn(\+V)$ and thus $\+V$ is a basis for $[\.g_\+A,\.g_\+A]$.

Consider $W=[A_{k_1},[A_{k_2},\ldots, [A_{k_{\ell-1}},A_{k_\ell}]\ldots]]$ for $\ell \ge 2$ and $A_{k_1},\ldots, A_{k_\ell}\in \tilde{\+A}$. We prove by induction on $\ell$ that $[W, V]\in \spn(\+V)$ for all $V\in\+V$, and thus $[\+V,\+V]\subseteq\+V$. For the base case $\ell=2$, by the Jacobi identity we have $[[A_{k_1},A_{k_2}],V]=[A_{k_1},[A_{k_2},V]] + [A_{k_2},[V,A_{k_1}]]\in\spn(\+V)$. Now consider $W = [A_{k_1},[A_{k_2}, \ldots, [A_{k_{\ell}}, A_{k_{\ell+1}}]]=[A_{k_1},W']$ for some $W'$ of length $\ell$. By the inductive hypothesis and the Jacobi identity we have $[W,V]=[[A_{k_1}, W'],V] = [A_{k_1},[W',V]]+[W',[V,A_{k_1}]] = [A_{k_1}, V'] + [W',V'']\in\spn(\+V)$ for some $V',V''\in\+V$. Thus,
\eq{
[\.g_\+A,\.g_\+A] &= \{[\alpha A + \beta V, \alpha' A' + \beta' V' ] : A,A'\in\tilde{\+A}, V,V'\in\+V, \alpha, \alpha', \beta,\beta'\in\mb{R}\} \\
&\subseteq\spn(\+V).
}
\end{proof}
}


\begin{lem}\label{lem:two_bases}
    Suppose $\{A_1, \ldots, A_L\}$ and $\{B_1, \ldots, B_L\}$ are two bases of the same subspace, then \[\langle \set{ A_1, \ldots, A_L} \rangle_{{{\rm Lie}, \R}} = \langle \set{B_1, \ldots, B_L} \rangle_{{{\rm Lie}, \R}}.\]
\end{lem}
\begin{proof}
    Each $B_j$ can be expressed as a linear combination of $\{A_k:k\in [L]\}$. Then anything generated by $B_j$ can be generated by $\{A_k:k\in [L]\}$ via the appropriate linear combination. Same for the other direction.
\end{proof}

\begin{lem}\label{lem:seed-center-gg}
    Any reductive Lie algebra $\.g = \langle \set{A_1, \ldots, A_L} \rangle_{{{\rm Lie}, \R}}$ generated by linearly independent generators $A_1, \ldots, A_L$ has
    \eq{
        \dim(\.c)+\dim(V_A \cap [\.g,\.g]) &=L,
    }
    where $V_A = \spn(\{A_1,\ldots ,A_L\})$ and  $\.c$ is the center of $\.g$.
\end{lem}
\begin{proof}
    Define subspace $V_1\defeq V_A \cap [\.g,\.g]$, and take any complementary subspace $V_2$ of $V_1$ in $V_A$, 
    i.e. $V_A=V_1+V_2$, $\dim(V_A)=\dim(V_1)+\dim(V_2)$ and $V_2 \cap [\.g,\.g] =\{0\}$. Take a basis $\{B_1,\ldots, B_c\}$ of $V_2$ and a basis $\{B_{c+1},\ldots,B_L\}$ of $V_1$. Note that $B_{c+1},\ldots, B_L \in V_1 \subseteq [\.g,\.g]$ by definition. Since $\{B_1,\ldots,B_L\}$ and $\{A_1,\ldots,A_L\}$ are two bases of the same space, by Lemma~\ref{lem:two_bases} they generate the same Lie algebra $\.g$. 
    Let $\{B_{L+1},..., B_D\}$ be a basis of vectors in $[\.g,\.g]$ for a complementary subspace of $V_A$ in $\.g$ (since $\.g$ is spanned by $\+A$ and all nested commutators of elements in $\+A$ it is always possible to choose such a basis).
    Now, since $\.g= \spn (\{B_1,\ldots,B_c,B_{c+1},\ldots,B_D \})= V_2+V_3$ where $V_3 = \spn(\{B_{c+1},\ldots,B_D \})$, $V_2\cap [\.g,\.g] = \{0\}$ and $V_3\subseteq[\.g,\.g]$, we have that $V_3=[\.g,\.g]$. (Indeed take any $v\in [\.g,\.g]$. Since ${\.g} = V_2 + V_3$, we can write $v$ as $v = v_2 + v_3$ for some $v_2\in V_2$ and $v_3\in V_3$. Then $v_2 = v - v_3\in [\.g,\.g]$. But $v_2\in V_2$ which intersects with $[\.g,\.g]$ only at 0, thus $v_2 = 0$ and therefore $v = v_3\in V_3$.) Thus $\dim([\.g,\.g])=\dim(V_3)=D-c$, and $\dim(\.c)+\dim(V_1)=D-\dim([\.g,\.g])+\dim(V_1)=c+\dim(V_1)=L.$
\end{proof}

\section{Cardinality-increasing generator set modifications}

We now prove Theorems~\ref{thm:extendgensetnaive}-\ref{thm:twogenerators}, which relate to DLA generator set modifications increasing the total number of generators.

\ExtendGenSetNaive*

 \begin{proof} Follows directly from Lemma~\ref{lem:projectors} \eqref{line:projectors-2}, \eqref{line:projectors-3} and Lemma~\ref{lem:copies_from_projectors}.
 \end{proof}

Recall that, for a dynamical generating set $\+A = \{A_1, \ldots, A_L\}$ and $S \subseteq [L],$ we denote by $\+A_S$ the set of operators 
$\{A_i \otimes I : i \in [L] \} \cup \{A_i \otimes \chi : i \in S\}$, where $\chi$ is a Hermitian operator.
For $1 \le i \leq L$,  we write $\+A_i$ for $\+A_{\{i\}}$. Since $A_i$ is traceless, so is $A_i\otimes \chi$, for all $i\in S$.

\begin{lem}\label{thm:all}

  For every dynamical generating set $\+A = \{A_1, \ldots, A_L\}$, we have 
 \eql{
  [\.g_{\+A_{[L]}}, \.g_{\+A_{[L]}}] = \spn( 
  [\.g_\+A,\.g_\+A]\otimes \+U_\chi)  
  \label{eq:A_L}
  }
  and
  \eq{
   Z(\.g_{\+A_{[L]}}) = 
   \spn(Z(\.g_\+A)\otimes \{\chi^0, \chi^1 \}).
  }
 \end{lem}

 \begin{proof}
   We start with the commutator ideal. By Proposition \ref{prop:generating-basis}, there exists a right-nested $\+A$-commutator basis $\+V = \{V_1, \ldots, V_D\}$ for $ [\.g_\+A,\.g_\+A]$.
    We set $\+B =  \+V \otimes \+U_\chi,$ by Lemma~\ref{lem:projectors} \eqref{line:projectors-1} it is a basis for 
    $\spn(  [\.g_\+A,\.g_\+A]\otimes \+U_\chi ).$
    We claim that $\+B$ is also a basis of $[\.g_{\+A_{[L]}}, \.g_{\+A_{[L]}}].$ 
    
     For the claim we first  prove that 
     $\+B \subseteq [\.g_{\+A_{[L]}}, \.g_{\+A_{[L]}}]$. For any $V\in\+V$, we will show by induction on $j$ that $V \otimes \chi^j \in [\.g_{\+A_{[L]}}, \.g_{\+A_{[L]}}],$
     for all $0 \leq j \leq K-1.$ 
    For $j=0$, the statement is true by the definition of
    $[\.g_{\+A_{[L]}}, \.g_{\+A_{[L]}}].$
    Now suppose that the statement is true for some
    $0 \leq j<K-1$ and let us show that
    $V \otimes \chi^{j+1} \in [\.g_{\+A_{[L]}}, \.g_{\+A_{[L]}}].$
    Since $V$
    is the nested commutators 
    of generators, it can be written in the form
    \eq{
V = [A_{k_{\ell}}, [ A_{k_{\ell -1}}, \ldots , [A_{k_{2}}, A_{k_{1}}]\ldots] ],
    }
    for some integers $\ell \geq 2$ and $k_1, \ldots, k_{\ell} \in [L]$. 
    Since $A_{k_1} \not\in Z(\.g_\+A)$ (otherwise the commutator would be $0$), by Fact~\ref{directsum} we can decompose it as $A_{k_1}
    = W+C$ where $W \in [\.g_\+A,\.g_\+A], W \neq 0$ and
    $C \in Z(\.g_\+A).$ Then we also have 
     \eq{
    V = [A_{k_{\ell}}, [ A_{k_{\ell - 1}}, \ldots , [A_{k_{2}}, W]\ldots] ],
    }
    and $W \otimes \chi^{j} \in [\.g_{\+A_{[L]}}, \.g_{\+A_{[L]}}]$
    by the inductive hypothesis.
    Also, $A_{k_2} \otimes \chi \in \.g_{\+A_{[L]}}$
    by assumption.
    Therefore $V \otimes \chi^{j+1}$ can be obtained as nested commutators of elements from $\.g_{\+A_{[L]}}$ as
    \eq{
    V \otimes \chi^{j+1} = 
    [A_{k_{\ell}}\otimes I, [ A_{k_{\ell -1}}\otimes I, \ldots , [A_{k_{2}} \otimes \chi,  W \otimes \chi^{j}] \ldots] ].
    }
    
    To finish the proof of the claim, observe that
    \eq{
    [\.g_{\+A_{[L]}}, \.g_{\+A_{[L]}}] \subseteq 
    \spn (\+V \otimes \{\chi^j : j \geq 0\}),
    }
    and note that, by Lemma~\ref{lem:projectors} \eqref{line:projectors-3}, the set 
    $\+V \otimes \{\chi^j : j \geq 0\}$
    is in the span of $\+B$.
    
    We now turn to the center of $\.g_{\+A_{[L]}}.$
Let $\+C$ be a basis for $Z(\.g_{\+A})$, then 
$\+C \otimes \{\chi^0, \chi^1\}$ is a basis of $\spn(Z(\.g_\+A)\otimes \{\chi^0, \chi^1 \}).$
We prove that it is also a basis of $Z(\.g_{\+A_{[L]}}).$
Lemma~\ref{decompositions} and the definition of $\+A_{[L]}$ 
imply that $\+C \otimes \{\chi^0, \chi^1\}$
is included in $Z(\.g_{\+A_{[L]}})$. Therefore it is sufficient to show that it also generates  $Z(\.g_{\+A_{[L]}})$.

Let $\sum_{j=0}^{K-1} C_j \otimes \chi^j $ be an arbitrary element
in $Z(\.g_{\+A_{[L]}}),$
where $C_j \in \.g_{\+A}$, for $0 \leq j \leq K-1.$ By 
Lemma~\ref{lem:center-from-original} \eqref{line:center-from-original-3}, we have $C_j \in Z(\.g_{\+A})$,
for $0 \leq j \leq K-1.$


    We now prove that for 
    $2 \leq j \leq K-1$, we also have 
    $C_j \in [\.g_{\+A},\.g_{\+A}]$.
    By Lemma~\ref{decompositions}, we can decompose
    $C_j \otimes \chi^j$  as 
    \eq{
      C_j \otimes \chi^j &= A+ W, 
     }
     where $A \in \spn( \+A_{[L]})$ and $W\in [\.g_{\+A_{[L]}},\.g_{\+A_{[L]}}].$
     Therefore $C_j \otimes \chi^j - A \in [\.g_{\+A_{[L]}},\.g_{\+A_{[L]}}].$ 
     Since $A \in \spn( \+A_{[L]})$, there exist $B_0, B_1 \in \.g_{\+A}$ such that
     $A = B_0 \otimes \chi^0 + B_1 \otimes \chi^1.$ Thus
     \eq{
     C_j \otimes \chi^j - B_1 \otimes \chi^1 - B_0 \otimes \chi^0 \in 
     [\.g_{\+A_{[L]}},\.g_{\+A_{[L]}}].
     }
     The subspaces $\{[\.g_\+A,\.g_\+A]  \otimes \chi^j : 0 \leq j \leq K-1\}$ are linearly independent and span 
     $[\.g_{\+A_{[L]}},\.g_{\+A_{[L]}}]$, thus for every $0 \leq j \leq K-1$,
     there exists a unique element $W_j \in [\.g_\+A,\.g_\+A]$ such 
     that 
     \eq{
     C_j \otimes \chi^j - B_1 \otimes \chi^1 - B_0 \otimes \chi^0 =
     \sum_{k=0}^{K-1} W_k \otimes \chi^k.
     }
     Therefore $C_j = W_j$ for all $j\ge 2$ and  $C_j \in [\.g_{\+A},\.g_{\+A}].$

     Putting this together, 
     we have $C_j \in Z(\.g_{\+A}) \cap [\.g_{\+A},\.g_{\+A}],$
     for $2 \leq j \leq K-1$.
     By Fact~\ref{directsum}, this implies that $C_j =0.$
     Therefore $C = C_0 \otimes \chi^0 + C_1 \otimes \chi^1$
     which finishes the proof.
 \end{proof}


Theorem~\ref{thm:extendgenset} then follows as a corollary of Lemma~\ref{thm:all}.
 
 \ExtendGenSet*


  \begin{proof}
      Lemmas~\ref{thm:all} and~\ref{lem:projectors} \eqref{line:projectors-2} imply that  
   \eq{ [\.g_{\+A_{[L]}}, \.g_{\+A_{[L]}}] = \spn( 
    [\.g_\+A,\.g_\+A]\otimes \+P_\chi) }   
      and then the statement follows immediately from Lemma~\ref{lem:copies_from_projectors}.
  \end{proof}

\begin{rem}For integers $2\le q \le K$, if one defines $\+A'_q=\{A_i \otimes \chi^j : i\in[L], 0\le j\le q-1\}$, then 
Theorem~\ref{thm:extendgensetnaive} corresponds to $\+A'_K$, while Theorem~\ref{thm:extendgenset} corresponds to $\+A'_2$.  By the same proof method as for Theorem~\ref{thm:extendgenset} one can interpolate between these two extremes, i.e., 
 \eq{
  [\.g_{\+A'_q}, \.g_{\+A'_q}]  \cong \bigoplus_{j=1}^{K}[\.g_\+A,\.g_\+A] \text{ ~~ and ~~ } 
  \dim(Z(\.g_{\+A'_q})) = q\dim(Z(\.g_{\+A'_q})),
  }
and achieve $K$ copies of $[\.g_\+A,\.g_\+A]$ and $q$ copies of the center.
\end{rem}

We now prove Theorem~\ref{thm:pauli}, which applies to Pauli DLAs. As any two Pauli operators on $n$-qubits either commute or anticommute, given a Pauli dynamical generating set $\+A$, one can define its corresponding \textit{anticommutation graph} to be $G=(V,E)$ where $V=\+A$ and $E=\{(u,v) : u,v\in\+A, uv = -vu\}$.


\PauliGenSet*


 \begin{proof}
    Our first claim is that for every $1 \leq j \leq L$, the generator
$A_j$ is in $[\.g_\+A, \.g_\+A].$ Indeed, let $A_k$ be a generator such that it anticommutes with $A_j$ . Then by the anticommutation property of Pauli operators, and fact that $A_k^2 = -I$, we have
$[A_{k}, [A_{k}, A_{j}]] \propto  A_{j}.$

Our second claim is that
 $A_j \otimes \chi \in  [\.g_{\+A_{i}}, \.g_{\+A_{i}}],$ for all $1 \leq j \leq L.$ The proof is by induction on the distance of $A_j$ from $A_i$ in the anticommutation graph of $\+A.$
    The statement is true for $A_i$ at distance $0$ from itself because
     \eq{
    [A_{k}\otimes I, [A_{k}\otimes I, A_{i} \otimes \chi]] \propto A_{i} \otimes \chi,
    }
    for any generator $A_k$ which anticommutes with $A_i.$  Now let $A_j$
    be of distance $d+1$ from $A_i$. By the inductive hypothesis there exists $k$ such that $A_{k}$ and $A_j$ anticommute and 
    $A_{k} \otimes \chi \in  \.g_{\+A_{i}}.$ Then  we have 
    \eq{
    [A_{k} \otimes \chi, [A_{k}\otimes I, A_{j}\otimes I]] \propto A_{j} \otimes \chi.
    }

The first claim implies that $[\.g_\+A, \.g_\+A] = \.g_\+A.$
The two claims together give 
$[\.g_{\+A_{i}}, \.g_{\+A_{i}}] = \.g_{\+A_{[L]}}.$
Since 
\eq{
[\.g_{\+A_{i}}, \.g_{\+A_{i}}] \subseteq \.g_{\+A_{i}} \subseteq\.g_{\+A_{[L]}}  {\rm ~~ and ~~}
[\.g_{\+A_{i}}, \.g_{\+A_{i}}] \subseteq 
[\.g_{\+A_{[L]}}, \.g_{\+A_{[L]}}] \subseteq\.g_{\+A_{[L]}},
}
we also have $\.g_{\+A_{i}} = [\.g_{\+A_{[L]}}, \.g_{\+A_{[L]}}]$.
The result then follows from
Theorem~\ref{thm:extendgenset}. 
\end{proof}


For generating sets consisting of only two generators, a similar result is obtained:

\CardTwoGenSet*


\begin{proof}
    We prove this for $i=1$. The $i=2$ case is identical except with $A_1,A_2$ interchanged. We start with the commutator ideal. Let $\+V = \{V_1, \ldots, V_D\}$ be a right-nested $\+A$-commutator basis for $ [\.g_\+A,\.g_\+A]$.
    That is, for all $V\in \+V$,
    \eql{
    V = [A_{k_{\ell}}, [ A_{k_{\ell - 1}}, \ldots , [A_{k_{2}}, A_{k_{1}}]\ldots]], \label{eq:nested}
    }
    for some integers $\ell \geq 2$ and $k_1, \ldots, k_{\ell} \in \{1,2\}$. 
    The indices $k_{1}$ and $k_{2}$ are different, and we can suppose without loss of generality that $k_{2} = 2$ and $k_{1} = 1$.  For every 
    $V \in \+V$, we show by induction that 
     $V \otimes \chi^j \in [\.g_{\+A_{1}}, \.g_{\+A_{1}}],$
     for all $1 \leq j \leq K-1$. The base case $j=0$ is trivial.
    
    Now suppose that the statement is true for some $1 \leq j<K-1$, and let us show that
    $V \otimes \chi^{j+1} \in [\.g_{\+A_{1}}, \.g_{\+A_{1}}]$. 
By Fact~\ref{directsum}, 
there exist
$C_1, C_2 \in Z(\.g_\+A)$ and $W_1, W_2 \in [\.g_\+A,\.g_\+A]$ such that
\eq{
A_i = C_i + W_i,
}
for $i \in\{1,2\}.$
Then 
 \eq{
    V = [W_{k_{\ell}}, [ W_{k_{\ell -1}}, \ldots , [W_{k_{2}}, A_{k_{1}}]\dots] ].
    }
By the inductive hypothesis $W_{k_{2}} \otimes \chi^j \in \.g_{\+A_{1}}$
and therefore 
\eq{
V \otimes \chi^{j+1} = [W_{k_{\ell}}\otimes I, [ W_{k_{\ell -1}}\otimes I, \ldots , [W_{k_{2}} 
\otimes \chi^j, A_{k_{1}} \otimes \chi]\dots] ]
}
is also in $[\.g_{\+A_{1}}, \.g_{\+A_{1}}] .$ 
Similar to the proof of Lemma \ref{thm:all}, it is easily shown that $[\.g_{\+A_{i}}, \.g_{\+A_{i}}] \subseteq 
    \spn (\+V \otimes \{\chi^j : j \geq 0\})$. Therefore, by Lemma~\ref{lem:projectors}, we have
\begin{align}\label{eq:2-generator-gg}
    [\.g_{\+A_{i}}, \.g_{\+A_{i}}] =
    \spn (\+V \otimes \{\chi^j : j \geq 0\}) = \spn (\+V \otimes \{\Pi_k : k\in [K]\}) \cong \bigoplus_{j=1}^K [\.g_\+A, \.g_\+A].    
\end{align}

\begin{table}
    \centering
    \begin{adjustbox}{width=0.9\textwidth}
    \small
    \begin{tabular}{|c|c|c|c|c|}
         \hline 
         & 
         $\spn(\+A) \cap [\.g_{\+A},\.g_{\+A}]$ & $\dim(Z(\.g_{\+A})) $ & $\spn(\+A_1) \cap [\.g_{\+A_1},\.g_{\+A_1}]$ & $\dim(Z(\.g_{\+A_1}))$\\ \hline 
         $A_1\in [\.g_{\+A},\.g_{\+A}]$, $A_2\in [\.g_{\+A},\.g_{\+A}]$  &  $A_1, A_2$ & 0 & $A_1\otimes I$, $A_2\otimes I$, $A_1\otimes \chi$ & 0 \\ \hline 
         $A_1\in [\.g_{\+A},\.g_{\+A}]$, $A_2\notin [\.g_{\+A},\.g_{\+A}]$ & $A_1$ & 1  & $A_1\otimes I$, $A_1\otimes \chi$ & 1 \\ \hline 
         $A_1\notin [\.g_{\+A},\.g_{\+A}]$, $A_2\in [\.g_{\+A},\.g_{\+A}]$  & $A_2$ & 1 & $A_2\otimes I$ & 2 \\ \hline 
         $A_1\notin [\.g_{\+A},\.g_{\+A}]$, $A_2\notin [\.g_{\+A},\.g_{\+A}]$ & $\{0\}$ & 2  & $\{0\}$ & 3 \\ \hline 
    \end{tabular}
    \end{adjustbox}
        \caption{Explicit calculation of $\dim(Z(\.g_{\+A_1}))$.     Entries in the second and fourth columns denote basis vectors for $\spn(\+A)\cap[\.g_\+A,\.g_\+A]$ and $\spn(\+A_1)\cap[\.g_{\+A_1},\.g_{\+A_1}]$, respectively.}
        \label{tab:proof-2-generators}
\end{table}
For the center of $\.g_{\+A_{1}}$, we divide the analysis into 4 cases depending on whether $A_1$ and $A_2$ lie in $[\.g_\+A,\.g_\+A]$. Results are summarized in Table \ref{tab:proof-2-generators}. The computation of a basis for $\spn(\+A) \cap [\.g_\+A,\.g_\+A]$ (second column) is elementary. 
From Eq.~\eqref{eq:2-generator-gg} and Lemma \ref{lem:projectors}, we see that $[\.g_{\+A_1}, \.g_{\+A_1}]= \{ h \otimes \chi^k : h\in[\.g_\+A,\.g_\+A],k=0,1,..., K-1\}$, and a basis of $\spn(\+A_1)  \cap [\.g_{\+A_1},\.g_{\+A_1}]$ (fourth column) then follows from the definition of $\+A_1$. The claim on the dimensions 
of the centers follows from Lemma \ref{lem:seed-center-gg}.

\end{proof}

\section{Cardinality-preserving generator set modifications}\label{sec:gensetcardincrease}

In this section we prove Theorem~\ref{thm:comcyc} and give a number of cases where it can be applied.

We say that a non-zero Hermitian operator $Q$ is \textit{sign unambiguous} if, for every  
pair of non-zero eigenvalues $\lambda,\lambda'$ in the spectrum of $Q$, we have $\lambda \neq -\lambda'$.

Let $\+A = \{A_1, \ldots, A_L\}$ be a dynamical generating set.
We say that $\+A$ is \textit{cyclic} if for every noncommuting 
$A_i, A_j \in \+A$, there exists a stable extension of $(A_i,A_j)$. That is, there exist $A_{k_1},\ldots,A_{k_\ell}$ s.t. $[A_{k_\ell},[\ldots,[A_{k_1},[A_j,A_i]]\ldots]] {\propto} [A_j,A_i]$. We say that an integer $M$ is a \textit{common cycle length} of 
a cyclic dynamical generating set $\+A$
if, for every noncommuting $A_i, A_j \in \+A,$ there exists a 
stable extension of $(A_i,A_j)$ 
of length $e(i,j)$ 
such that $M$ is the least common multiple
of $\{ e(i,j) : [A_i,A_j] \neq 0\}.$ 
 We call a dynamical Lie algebra 
$\.g$ \textit{cyclic} if there exists a cyclic dynamical
generating set $\+A$ such that $\.g =\.g_\+A$. 


\begin{rem}
    If a dynamical generating set is cyclic, this implies a more general property:  given any nested commutator basis $\+V$ for $[\.g_\+A,\.g_\+A]$, for $V\in\+V$, there exist $\+A$-sequences $\underline{A},\underline{A'}$ such that (i) $V = \val({\underline{A}}\circ\underline{A'})$; and (ii) there exists a stable extension of $\underline{A}$.  Theorem~\ref{thm:comcyc} below can be adapted to apply to this more general case, but we state the theorem in terms of the stronger cyclic property as in all of our applications of the theorem the dynamical generating sets satisfy this.
\end{rem}

\CardGenSetFixed*


\begin{proof} 
We first consider the $[\.g_{\+A_{ Q}},\.g_{\+A_{ Q}}]$ part.  Let $M$ be a common cycle length of $\+A$ and 
let $\+V = \{V_1, \ldots, V_D\}$ be a nested commutator basis for $[\.g_\+A,\.g_\+A]$. For any  $V \in \+V$ there exists an $\+A$-sequence
\eq{
\underline{A} =(A_{k_{1}},  A_{k_{2}}, \ldots , A_{k_{\ell -1}}, A_{k_{\ell}} )
}
such that  $\ell \geq 2$ and $\val({\underline{A}}) = V.$ By the definition of a common cycle length,
there exists an $\+A$-sequence $\underline{A}(k_1,k_2)$ of length $e(k_1,k_2)$
such that $\start(\underline{A}) \circ \underline{A}(k_1,k_2)$ is a stable extension of $\start(\underline{A})$ and $e(k_1,k_2)$ divides $M.$

For every integer $j \geq 0,$ we set
\eq{
\underline{A}(k_1,k_2, j) = \underline{A}(k_1,k_2) \circ \cdots \circ \underline{A}(k_1,k_2),
}
where $jM/e(k_1,k_2)$ copies of $\underline{A}(k_1,k_2) $ are concatenated. Then 
$\underline{A}(k_1,k_2, j)$ is of length $jM$ and 
$\start(\underline{A}) \circ \underline{A}(k_1,k_2,j)$ is also a stable extension
of $\start(\underline{A})$. From this we can conclude that
$\start(\underline{A}) \circ \underline{A}(k_1,k_2,j) \circ \fin(\underline{A})$
is of length $\ell + jM$ and 
\eq{
\val(\start(\underline{A}) \circ \underline{A}(k_1,k_2,j) \circ \fin(\underline{A}))
\propto V.
}

This in turn implies that,  if we define $\chi = Q^M$, then
for every $V\in\+V$, 
 there exists $\ell,$ such that we have
\eq{
V\otimes Q^{\ell + jM} 
\in [\.g_{\+A_{G,Q}},\.g_{\+A_{G,Q}}],
}
for all $j\ge 0$, because we can allocate one $Q$ to each $A_{k_j}$ in $\start(\underline{A})$, each $A_{k_j}$ in $\fin(\underline{A})$, and each $A_{k_j}$ in each copy of $\underline{A}(k_1,k_2)$.

Let $Q = \sum_{i=1}^K \lambda_i \Pi_i$ be the spectral decomposition of $Q$ 
where  $\Pi_i$ is the orthogonal projector onto the $\lambda_i$ eigenspace of $Q$, for $i\in[K]$. 
We claim that $\spn(\{Q^{\ell + jM}  : 0 \leq  j \leq K-1 \})
=\spn(\{\Pi_1, \ldots, \Pi_K\})$. To see that, consider the equation
\eq{
\bpm Q^{\ell} \\ Q^{\ell + M} \\ \vdots \\ Q^{\ell + (K-1)M} \epm &= 
\bpm 
\lambda_1^{\ell} & \lambda_2^{\ell}  & \cdots & \lambda_K^{\ell}  \\
\lambda_1^{\ell + M} & \lambda_2^{\ell + M} & \cdots & \lambda_K^{\ell + M} \\
\vdots & \vdots & \ddots & \vdots \\
\lambda_1^{\ell + (K-1) M} & \lambda_2^{\ell + (K-1) M} & \cdots & \lambda_K^{\ell + (K-1) M}
\epm
\bpm
\Pi_1 \\ \Pi_2 \\ \vdots \\\Pi_K
\epm.
}

The determinant of the matrix in the equation is 
$(\prod_{j=1}^K \lambda_j)^\ell$ times the determinant of the Vandermonde matrix $V(\lambda_1^M, \ldots, \lambda_K^M).$
Since $Q$ is sign unambiguous and all $\lambda_i$ are distinct, so are all $\lambda_i^M$. Therefore the above determinant is non-zero,
from which the claim follows
using Lemma~\ref{lem:projectors}, and we get
\begin{align}\label{eq:AQ-gg}
    [\.g_{\+A_{Q}},\.g_{\+A_{Q}}] =
\spn(\{V \otimes \Pi_j: V\in\+V, 1 \leq j \leq K\}).
\end{align}
By Lemma~\ref{lem:copies_from_projectors} this gives the isomorphism 
\eq{
[\.g_{\+A_{Q}},\.g_{\+A_{Q}}]\cong \bigoplus_{j=1}^K [\.g_{\+A},\.g_{\+A}].
}

For the center, clearly $C\otimes Q \in Z(\.g_{\+A_{Q}})$, for every
$C \in Z(\.g_{\+A_{Q}})$. Indeed, $C\otimes Q$ commutes with all generators, and thus also commutes with $\.g_{\+A_Q}$. It also holds that $C\otimes Q\in \.g_{\+A_Q}$: $C$ can be written as a linear combination of $A_j$ plus some $B\in [\.g_{\+A},\.g_{\+A}]$. All $A_j\otimes Q$ are in $\.g_{\+A_Q}$ by definition, and $B\otimes Q$ is also in $\.g_{\+A_Q}$ because of Eq.~\eqref{eq:AQ-gg}.

The proof of the other inclusion is basically
identical to the proof of the analogous claim in Lemma~\ref{thm:all}. 
Let $\sum_{j=0}^{K-1} C_j \otimes Q^j $ be an arbitrary element
in $Z(\.g_{\+A_{Q}}),$
where $C_j \in \.g_{\+A}$, for $0 \leq j \leq K-1.$ By 
Lemma~\ref{lem:center-from-original} \eqref{line:center-from-original-3}, we have $C_j \in Z(\.g_{\+A})$,
for $0 \leq j \leq K-1.$

    We claim that for 
    $j \neq 1$, we also have 
    $C_j \in [\.g_{\+A},\.g_{\+A}]$, which 
     by Fact~\ref{directsum} implies that $C_j =0.$ To prove the claim, 
    by Lemma~\ref{decompositions}, we decompose
    $C_j \otimes Q^j$  as 
    \eq{
      C_j \otimes Q^j &= A+ W, 
     }
     where $A \in \spn( \+A_{Q})$ and $W\in [\.g_{\+A_{Q}},\.g_{\+A_{Q}}].$
     Therefore $C_j \otimes Q^j - A \in [\.g_{\+A_{Q}},\.g_{\+A_{Q}}].$ 
     Since $A \in \spn( \+A_{Q})$, there exist $B \in \.g_{\+A}$ such that
     $A =  B \otimes Q.$ Thus
     \eq{
     C_j \otimes Q^j - B \otimes Q  \in 
     [\.g_{\+A_{Q}},\.g_{\+A_{Q}}].
     }
     Then there exist unique elements $W_k \in [\.g_\+A,\.g_\+A]$ such 
     that 
     \eq{
     C_j \otimes Q^j - B \otimes Q =
     \sum_{k=0}^{K-1} W_k \otimes Q^k.
     }
     Therefore $C_j = W_j$ for all $j\ne 1$, i.e., $C_j \in [\.g_{\+A},\.g_{\+A}]$.
\end{proof} 

We now turn to a number of applications of Theorem~\ref{thm:comcyc}. Let $G=(V,E)$ be a graph on $n$ vertices. 
We define the set of QAOA-MaxCut DLA generators to be $\+M_G=\{A_1,A_2\}$, where 
\eq{
A_1&= \sum_{j=1}^n iX_j, \quad A_2=\sum_{(j,k)\in E} iZ_jZ_k.
}
Here, $A_1,A_2$ are $n$-qubit operators, $X_j$ is the Pauli $X$ operator acting on the $j$-th qubit (with the identity acting on all others) and, $Z_jZ_k$ is the tensor product of Pauli $Z$ operators acting on qubits $j$ and $k$ (again, with the identity acting on all other qubits). 
Similarly, we define $\+S_G=\{A_1,A_2, A_3\}$, where
\eq{
A_1 &= \sum_{j=1}^n iX_j,\quad  A_2= \sum_{j=1}^n iY_j,\quad  A_3 = \sum_{(j,k)\in E} iZ_jZ_k.
}

\begin{rem}
    The special case of $\+S_G$ where $G=K_n$ corresponds to the so-called $S_n$-equivariant dynamical Lie algebra studied by~\cite{albertini2018controllability,schatzki2024theoretical}.
\end{rem}

\begin{thm}\label{thm:qaoa} For any graph $G$, the dynamical generating sets $\+M_{G,Q}=\{A \otimes Q : A \in \+M_{G}\}$ and $\+S_{G,Q}=\{A \otimes Q : A \in \+S_{G}\}$ are cyclic, with common cycle length $2$.



\end{thm}

\begin{proof}
By explicitly computing the appropriate nested commutators, we can verify that in  $\.g_{\+M_{G,Q}}$ we have
\eq{
[A_1, [A_1,[A_1, A_2]]]& \propto [A_1, A_2],
}
and that in $\.g_{\+S_{G,Q}}$ we have
\eq{
[A_1, [A_1, [A_1, A_3]]] &\propto  [A_1,A_3], \\
[A_2, [A_2, [A_2, A_3]]] & \propto  [A_2,A_3], \\
[A_2, [A_2, [A_1, A_2]]] & \propto  [A_1,A_2].
}
\end{proof}

\begin{thm}\label{generalpauli}
Let $\+A$ be a dynamical generating set
such that for all $A \in \+A$, there exists $\lambda_A < 0$ for which $A^2 = \lambda_A I$. Then, $\+A$ is cyclic with common cycle length at most $2$. 
\end{thm}
\begin{proof}  
We claim that for all $A,B \in \+A$, we have
\eq{
[A,[A,[A,B]]] = 4\lambda_A[A,B].
}

For any operators $A,B$ it is straightforward to verify that $[A^n,B]=A[A^{n-1},B] - [A,[A^{n-1},B]] + A^{n-1}[A,B]$. From this it follows that $[A^2,B]=2A[A,B]-[A,[A,B]]$ and
\eq{
[A^3,B]&=A[A^{2},B] - [A,[A^2,B]] + A^2[A,B] \\
&= A(2A[A,B] - [A,[A,B]])-[A,2A[A,B]-[A,[A,B]]] + A^2[A,B] \\
&=2A^2[A,B] -A[A,[A,B]]-2[A,A[A,B]]+ [A,[A,[A,B]] + A^2[A,B] \\
&= 3A^2[A,B] + [A,[A,[A,B]]  -3A[A,[A,B]].
}
Thus,
\eq{
[A,[A,[A,B]]]&= 3A[A,[A,B]] + [A^3,B]-3A^2[A,B] \\ 
 &=3A(2A[A,B]-[A^2,B]) + [A^3,B]-3A^2[A,B]\\
&=3A^2[A,B] + [A^3,B]-3A[A^2,B].
}
If $A^2=\lambda_A I$ then $[A,[A,[A,B]]] = 3\lambda_A [A,B] + \lambda_A [A,B] = 4\lambda_A [A,B]$. 
\end{proof}

\begin{rem} If $\+A=\{A_1,A_2\}$ 
then Theorem~\ref{generalpauli} already holds 
when either $A_1^2 = \lambda I$ or $A_2^2 = \lambda I$, for some $\lambda <0$. 
\end{rem}

\begin{thm}\label{thm:pauli_two} Every Pauli dynamical generating set is cyclic, with common cycle length $2$.
 \end{thm}

\begin{proof}
Immediate from Theorem~\ref{generalpauli},  since elements of a Pauli dynamical generating set have eigenvalues $\pm i$.
\end{proof}

From Theorem~\ref{thm:comcyc}, we obtain the following corollary:

\begin{cor} Let $\+A$ be a dynamical generating set from Theorems~\ref{thm:qaoa},\ref{generalpauli},or~\ref{thm:pauli_two}, and let $Q$ be a sign unambiguous
 Hermitian operator with $K$ distinct non-zero eigenvalues.
Set $\+A_{Q}=\+A\otimes \{Q\}$. Then, 
\eq{
[\.g_{\+A_{ Q}},\.g_{\+A_{ Q}}]\cong \bigoplus_{j=1}^{K} [\.g_{\+A},\.g_{\+A}]
,
}
and 
\eq{
 Z(\.g_{\+A_{Q}}) =   \{C\otimes Q : C \in Z(\.g_{\+A})\}.
}   
\end{cor}

\section{Discussion}

In this work we have initiated a study of qubit- and parameter-efficient ways in which a DLA generating set $\+A$ can be modified so that the resulting DLA is isomorphic to a direct sum of multiple copies of $\.g_\+A$. For variational quantum algorithm design, it would also be desirable to have modifications to $\+A$ effecting other DLA transformations. For example we have the following related questions:
\begin{itemize}
    \item In a similar spirit to the work of~\cite{zimboras2015symmetry}, under what conditions would the change $\+A\ra\+A'$ leave the DLA invariant?
    \item Given a specific subalgebra $\.h \subseteq \.g_\+A$, what change $\+A\ra\+A'$ would give $\.g_{\+A'}\cong \.h$?
    \item Given generator sets $\+A$ and $\+B$, what generator sets $\+C$ would give $\.g_\+C \cong \.g_\+A \oplus \.g_\+B$?
\end{itemize}
As is the case for generating tensor powers of a DLA, it is easy to construct trivial ways to achieve each of the above. However, modifying the generator sets in a way that is either qubit- or parameter-efficient, or else can be used to analyze the structure of other DLAs, may be more involved. But, if such tools can be developed, they would greatly aid the study and design of variational quantum algorithms where DLAs, other than those with Pauli generators, remain difficult to analyze in general.

\section*{Acknowledgements}
{M.S. was} supported by the National Research Foundation, Singapore through the National Quantum Office, hosted in A*STAR, under its Centre for Quantum Technologies Funding Initiative (S24Q2d0009).

\normalem
\bibliographystyle{unsrt}
\bibliography{template.bib}

\appendix
\section{Proof of Fact~\ref{fact:nested}}\label{app:nested-comms}

Here we give a proof of Fact~\ref{fact:nested} from the main text.

Let $\+A$ be a dynamical generating set. Define the sets $\+S_1=\+A$, $\+S_2=\{[A,B]: A,B\in \+A\}$ and, for integers $j\ge 3$, 
\eq{
\+S_{j}&=\bigcup_{{k=1}}^{j-1} \{[A,B] : A \in \+S_{k}, B\in \+S_{{j-k}} \}.
}
In other words, $\+S_n$ is the set of nested commutators of $n$ elements from $\+A$, where the nesting occurs in any order.  For $k\ge 2$, define right-nested commutators of degree $k$ to be nested commutator of the form $[A_{i_{1}},[\dots[A_{i_{k-1}},A_{i_k}]\dots]]$.  

\begin{lem}\label{lem:nested_any_two} 
Let $P$ be an element of $\+A$ or a right-nested commutator of degree $\ge 2$, and $Q$ a right-nested commutator of degree $2$.  Then, $[P,Q]$ can be expressed as a linear combination of right-nested commutators.
\end{lem} 
 
\begin{proof}
Write $Q=[A,B]$ for some $A,B\in\+A$. Then, by the Jacobi identity and anticommutativity of the Lie bracket,
\begin{align}
[P,[A,B]]&= -[A,[B,P]] - [B,[P,A]] = -[A,[B,P]] + [B,[A,P]].
\end{align}
\end{proof}

\begin{lem}\label{lem:nested_any_any}
Let $P$ be an element of $\+A$ or a right-nested commutator of degree $\ge 2$, and $Q$ a right-nested commutator of degree $\ge 2$. Then, $[P,Q]$ can be expressed as a linear combination of right-nested commutators. 
\end{lem}

\begin{proof}
The proof is by induction, with the hypothesis that the claim is true for $P$ an element of $\+A$ or a right-nested commutator of degree $\ge 2$
and $Q$ a nested commutator of degree $\ell$.  By Lemma~\ref{lem:nested_any_two}, the base case $\ell=2$ is true.
  
Write $Q=[A,Q']$ where $A\in \+A$ and $Q'$ is a degree-$\ell$ right-nested commutator. Then,
\begin{align}
[P,Q] &= [P,[A,Q']] = - [A,[Q',P]]-[Q',[P,A]] \quad \text{(Jacobi)}\\
&=[A,[P,Q']]-[[A,P],Q']
\end{align}

By the inductive hypothesis, $[P,Q']$ is a linear combination of right-nested commutators, and thus so is $[A,[P,Q']]$. 
$[A,P]$ is a right-nested commutator, and again by the inductive hypothesis $[[A,P],Q']$ is a linear combination of right-nested commutators.
\end{proof}

We now prove Fact~\ref{fact:nested}, which we restate as follows:
\begin{thm}
For $n\ge 2$, any $X\in \+S_n$ can be written as a linear combination of right-nested commutators.
\end{thm}

\begin{proof}
We induct on $n$. The base case $n=2$ is trivial. Any $X\in \+S_{n+1}$ is of the form $X=[X',X'']$ where $X'\in \+S_{j}$ and $X''\in \+S_{n+1-j}$, for some $j\in[n]$. By the inductive hypothesis, both $X',X''$ can be written as linear combinations of right-nested commutators. The claim then follows from Lemma~\ref{lem:nested_any_any}.
\end{proof}

\end{document}